\documentclass[acmsmall]{ec20acm}

\usepackage{wrapfig}

\usepackage{booktabs} 

\usepackage{amsmath}
\usepackage{amssymb}
\usepackage{dsfont}

\usepackage{breqn}

\usepackage[ruled]{algorithm2e} 

\usepackage{graphicx}

\usepackage{url}
\urldef{\mailgr}\path|geoff.ramseyer@cs.stanford.edu|
\urldef{\mailag}\path|ashishg@stanford.edu|

\graphicspath{ {.} }

\copyrightyear{2020}
\acmYear{2020}
\setcopyright{acmlicensed}\acmConference[EC '20]{Proceedings of the 21st ACM Conference on Economics and Computation}{July 13--17, 2020}{Virtual Event, Hungary}
\acmBooktitle{Proceedings of the 21st ACM Conference on Economics and Computation (EC '20), July 13--17, 2020, Virtual Event, Hungary}
\acmPrice{15.00}
\acmDOI{10.1145/3391403.3399533}
\acmISBN{978-1-4503-7975-5/20/07}

\title{Continuous Credit Networks and Layer 2 Blockchains: Monotonicity and Sampling}
\author{Ashish Goel}
\affiliation{%
  \institution{Stanford University}
  \city{Stanford}
  \state{California}
  \postcode{94305}
}
\email{ashishg@stanford.edu}
\author{Geoffrey Ramseyer}
\affiliation{%
  \institution{Stanford University}
  \city{Stanford}
  \state{California}
  \postcode{94305}
}
\orcid{0000-0003-3205-2457}
\email{geoff.ramseyer@cs.stanford.edu}

\begin{abstract}
  To improve transaction rates, many cryptocurrencies have implemented so-called ``Layer-2'' transaction protocols, where payments are routed across networks of private payment channels.  However, for a given transaction, not every network state provides a feasible route to perform the payment; in this case, the transaction must be put on the public ledger.  The payment channel network thus multiplies the transaction rate of the overall system; the less frequently it fails, the higher the multiplier.

  We build on earlier work on credit networks and show that this network liquidity problem is connected to the combinatorics of graphical matroids.  Earlier work could only analyze the (unnatural) scenario where transactions had discrete sizes.
  In this work, we give an analytical framework that removes this assumption.  This enables meaningful liquidity analysis for real-world parameter regimes.

  Superficially, it might seem like the continuous case would be harder to examine.  However, removing this assumption lets us make progress in two important directions.  First, we give a partial answer to the ``monotonicity conjecture'' that previous work left open. This conjecture asks that the network's performance not degrade as capacity on any edge increases.  And second, we construct here a network state sampling procedure with much faster asymptotic performance than off-the-shelf Markov chains ($O(\vert E\vert \beta(\vert E\vert))$, where $\beta(x)$ is the complexity of solving a linear program on $x$ constraints.)

  We obtain our results by mapping the underlying graphs to convex bodies and then showing that the liquidity and sampling problems reduce to bounding and computing the volumes of these bodies. The transformation relies crucially on the combinatorial properties of the underlying graphic matroid, as do the proofs of monotonicity and fast sampling.

\end{abstract}



\begin{document}

\maketitle

\section{Introduction}

\subsection{Problem Motivation}
Recent years have witnessed a dramatic rise in the importance of cryptocurrencies such as Bitcoin \cite{nakamoto2008bitcoin} and Ethereum \cite{buterin2014next} as means of exchanging money and storing value \cite{coinbasebtcgrowth}.  Unfortunately, however, most blockchains can directly process only a small number of transactions per second.  Most Bitcoin variants, for example, only support 4-5 transactions per second \cite{bitcoinscalinghackernoon}.  One method for improving data throughput is to settle transactions privately, using the blockchain as fallback security guarantor.  This enables the construction of so-called ``off-chain'' or ``Layer-2'' networks. The most famous of these networks is the Lightning network on Bitcoin\cite{poon2016bitcoin}, but Layer 2 protocols have also been implemented in Ethereum and other cryptocurrencies \cite{raiden,connext,outpace}.  

The core idea of these networks is that if two parties frequently transact, they need not record every transaction on the global ledger, but rather, need only record the net total of their transactions at the end of a business relationship.  Traditional commercial banks behave in a similar manner; instead of repeatedly transferring cash, they often have accounts with each other, and simply move money in and out of these accounts.

Traditional banks can go to court if a financial transfer never arrives.  In an anonymous cryptocurrency network, or when blockchains span political jurisdictions, users may have no such recourse.  Instead, to guarantee solvency, two parties can agree to put money into escrow.  Transferring money then means privately updating shares of ownership of the escrowed funds.  Such an escrow account is typically called a ``payment channel.''  Pairs of agents who lack a direct channel can route funds across a path in the payment channel graph.  Figure \ref{fig:ln} gives an example payment network on 5 nodes.
\begin{wrapfigure}{r}{0.4\textwidth}
  \centering
  \caption{An example payment network on 5 nodes.  In this network, $A$ and $B$ have together invested 10 units of capital towards the edge $(A,B)$.  From this configuration, $A$ can send up to $5$ units of money to $B$.  $B$ could send up to $4$ units of money to $D$ through the graph, $2$ units directly and $2$ via $C$.}
  \includegraphics[width=5.59cm]{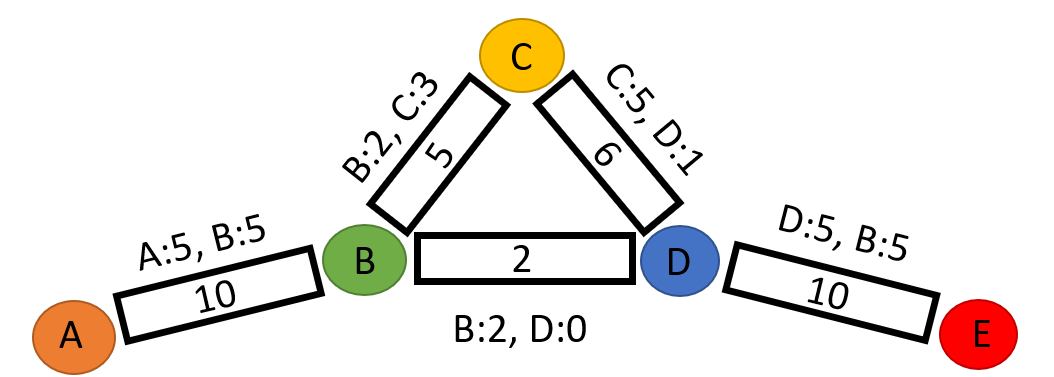}
  \label{fig:ln}

\end{wrapfigure}

However, if one agent comes to own the entirety of the escrowed funds, it cannot accept additional money on that channel without opening itself to the risk of being cheated by its channel counterparty.  More generally, no fixed network of payment channels can satisfy every possible sequence of transactions.  We give a more detailed description of the network model, based on credit networks \cite{dandekar2011liquidity}, in Section 2.

When a payment channel network cannot process a transaction, that transaction falls back to the underlying blockchain.  The ratio of the number of transactions successfully settled privately to the number of transactions that a payment channel network cannot settle thus acts as a multiplier on a blockchain's limited transaction rate.  This ratio is our object of study; a higher ratio is clearly desirable for improving blockchain performance.  Higher liquidity can also benefit individual agents directly, in that a higher transaction throughput capacity should mean lower fees per individual transaction, and higher liquidity means a higher chance of being able to purchase something online quickly.

Implementing a payment channel network in the real world is a complex systems design problem.  For example, agents must have some way of discovering payment routes across the network.  Some users might also like to keep their payment patterns private from network intermediaries.  However, as we show here, even if a payment network operates perfectly -- no hardware ever fails, every node has full information, transactions resolve instantaneously -- no payment network can resolve every transaction (under a natural class of transaction models).  Therefore, the baseline performance of e.g. a payment routing system is not 100\% success, but rather, the transaction success rate in an ideal system.  Our goal in this work is to study this baseline performance.

In real-world implementations of payment channel networks, off-chain transactions can settle in seconds, but putting data onto the blockchain or settling on-chain can take 10s of minutes.  Because the network's topology itself cannot change fast enough to respond to individual transactions, we study the performance of fixed network topologies.

Of course, placing capital into escrow to secure a link incurs a cost, in terms of the interest rate on the escrowed funds.  When agents participate in a payment network, they implicitly choose to pay this cost in exchange for the ability to use the network directly.  One goal of this work is to provide a framework for agents to understand the utility they would derive from participating in the network in a particular manner, and thus to evaluate whether this utility outweighs the costs they must pay.  Agents must make a strategic choice about how to allocate their escrowed funds.  One agent could choose to make one high-capacity link with a central node, or it could choose to make several lower-capacity links with different nodes.  An agent might wish to avoid making a link with another agent whose internet connection is unreliable.  Given some set of feasible options, agents must implicitly pick a preferred option.  One goal of this work is to lay out a framework for agents to evaluate and compare their options.

\subsection{Monotonicity}

At a high level, a payment network appears similar to a single-commodity flow network when resolving a single transaction.  The main difference is that transactions have side effects, and side effects of different transactions can cancel each other out or compound each other.  Many network flow problems or queueing problems study the case where the flow across an edge in each direction is independently limited.  In a payment network, flow in one direction cancels out flow in the other direction.  At any point in time, on each link, the total flow sent since the start of network operation in one direction must be approximately equal to the total flow in the other direction (where approximately means $\pm$ the amount of funds in escrow on that link).

Many seemingly similar network flow problems are subject to what has become known as Braess's paradox \cite{braess2005paradox}.  Famously, in car traffic networks, when drivers choose their routes selfishly, total traffic throughput in a road network can be decreased by adding roads.  Example 4.6 of \cite{ramseyer2020liquidity} gives a simple example of a credit network variant that violates monotonicity.  Many network models display similar deleterious effects.  

In cryptocurrency networks, where no central party exists to act as a coordinator, if selfishly routing payments had this deleterious effect, there would be no obvious way to coordinate a more efficient routing process.  It could also be possible for an adversarial party to add edges to the network to decrease throughput (there might even be a financial incentive to do this -- offchain transactions reduce demand for data on the blockchain, which lowers on-chain transaction fees).
There naturally arises, then, a question about the ``monotonicity'' of payment networks -- is the fraction of transactions that succeed monotonically non-decreasing as the number and capacity of edges increase?

A Layer-2 payment network closely corresponds to a distributed currency model known as a credit network, first formulated independently by \cite{ghosh2007mechanism,defigueiredo2005trustdavis,karlan2009trust}.  Dandekar et. al. and Goel et. al. \cite{dandekar2011liquidity,goel2015connectivity} developed the mathematical foundations for the model and analyze transaction success rates, which they call ``liquidity''.  In order to perform their analysis, these authors assume that transactions have discrete sizes.  With this assumption, the monotonicity conjecture is equivalent to a conjecture that matroid theory's ``negative correlation'' conjecture holds on the collection of independent sets of a graphical matroid.  Feder and Mihail in \cite{as1992balanced} showed that if a matroid has this property, one can sample in polynomial time an approximately random basis, but not all matroids have this property.  More recently, Anari et. al. \cite{anari2018log} gave a polynomial-time approximate sampling algorithm for general matroids.

\subsection{Our Results}

In this work, we give a general analytical framework for Layer-2 payment networks with arbitrarily sized transactions, eliminating a key assumption needed for \cite{goel2015connectivity,dandekar2011liquidity}, and use this toolkit to prove the monotonicity conjecture for the classes of transactions most relevant to cryptocurrency applications.  We also give an efficient network state sampling algorithm.  

The detailed results are more understandable with the following two pieces of context.  

First, formally defining the ``probability that a transaction is successful'' requires a small amount of care.  We will define it as the probability that a given transaction is feasible in configuration sampled according to some probability measure $\mu$ on the set of configurations.  This, of course, raises the question of what measures $\mu$ arise from real-world behavior.  

We model transactions as arising from an exogenous random process.  E-commerce transactions, for example, arise largely from demand for consumer goods, not from anything related to the short-term status of a payment processing pipeline.  For consistency with prior work \cite{goel2015connectivity,dandekar2011liquidity}, we also model transactions as arising in sequence in discrete time.   At every timestep, a payment network simply executes the current transaction if it is able in its current escrow-configuration (where a network's escrow-configuration is the ownership status of every edge; henceforth ``configuration'').  Every successful transaction changes the configuration of the network, so a random transaction model induces a Markov chain on the set of configurations of the network.   In the rest of this work, we measure transaction success probability relative to the invariant measure of this Markov chain.  There are other similar transaction models that would define natural, induced probability measures; results derived through our analytical framework depend on the probability measure, not the process through which a measure is induced.


Second, the state space of the induced Markov chain is most succinctly described not as a set of escrow-configurations but rather, as a set of classes of functionally-equivalent configurations.  An agent can send money to itself along a cycle to change the network configuration -- however, every transaction realizable in the network configuration prior to the cyclic flow is realizable in the network configuration after the cyclic flow.  Sending flow along cycles gives a ``cycle-equivalence'' relation between configurations, and for the purpose of executing transactions, representatives of one equivalence class are all functionally equivalent.  The induced Markov chain is thus most naturally viewed as acting on the equivalence classes of this equivalence relation.

Prior work primarily analyzed the case where transaction distributions are symmetric, which is to say, agent $u$ is as likely to pay agent $v$ as $v$ is likely to pay $u$.  Such transaction distributions induce uniform probability measures.  This assumption may not be realistic in some real-world credit network use cases.  For clarity, note that of the below results, result 1 (sections 2-4) does not rely on this assumption, while result 2 (section 5) does, as does Theorem \ref{thm:uniform_sample_algorithm} of result 3.  We give a small comment on this assumption and related open problems in section 7.  

Our results are the following:

\begin{enumerate}

\item

In sections 2, 3, and 4, we construct a new toolset for analyzing the set of configurations of a credit network.  Simply put, we construct a map $D$ (for ``Direction'' of a transaction) from the set of possible transactions to $\mathbb{R}^n$ and a map $S$ from the set of configurations of a credit network to $\mathbb{R}^n$ that together preserve how transactions and configurations interact.  If a transaction $\tau$ executed from configuration $w_1$ sends the network into configuration $w_2$, then $S(w_1)+D(\tau)=S(w_2)$.  The stationary measure of the induced Markov chain is a probability measure on $\mathbb{R}^n$, and analyzing liquidity reduces to analyzing the measures of subsets of $\mathbb{R}^n$.

For comparison to prior work, we study a continuous extension of the credit network model of \cite{goel2015connectivity,dandekar2011liquidity}.  Specifically, Goel et. al. and Dandekar et. al. \cite{goel2015connectivity,dandekar2011liquidity} measure performance with regard to transactions of size $1$, where $1$ is some minimum transaction size, while we study transactions of arbitrary real-valued size.  

If transactions all have the same size (or integral sizes), then configurations of the credit network can be mapped into forests (acyclic subsets of edges) of a related, undirected multigraph with the same vertex set.  They show that the probability that a transaction of size $1$ succeeds from agent $u$ to agent $v$ is equal to the probability that a forest sampled uniformly at random puts $u$ and $v$ in the same connected component.  

Unfortunately, this correspondence gives no clues as to how to analyze liquidity for transactions of larger size.  In real-world blockchains like Bitcoin, one Satoshi -- one ``smallest unit of money'' -- is worth at the time of writing roughly 1/10000th of a US dollar, while Lightning channel capacities may run into the hundreds or thousands of dollars.  In this kind of parameter regime, a transaction of size $1$ succeeds with probability very close to $1$ in any network topology, and most transactions sent have much larger sizes.  An agent in the network would be interested primarily in transactions of larger, but variable, sizes.  

We also study the rate at which transaction success rates decline as transaction sizes increase.  Modeling transactions as taking on arbitrary real-valued sizes streamlines this analysis and requires a new set of analytical tools.

As an aside, these analytical frameworks are in fact related mathematically.  From the point of view of prior work, we study the success rate of transactions of fixed real-world value, as discretization becomes arbitrarily fine.  Informally, analyzing the discrete case compared to the continuous case is like studying the number of points of a lattice within a well-behaved convex body, instead of directly studying the volume of said convex set.  Prior work of \cite{goel2015connectivity,dandekar2011liquidity} is able to measure liquidity using forests because of a correspondence between forests and these lattice points (\cite{stanley1980decompositions}, Example 3.1).  As a lattice becomes increasingly fine, a (weighted) count of the number of lattice points in a well-behaved convex set approximates the volume of said set; we note without proof that many of our arguments could be adapted to hold if the transaction granularity is sufficiently small using this principle.

\item

In section \ref{sec:monotonicity}, we prove theorems about the liquidity of general credit networks and a limited version of the monotonicity conjecture.

Relative to a uniform probability measure (i.e. transaction distributions that are symmetric), transaction success probability can be computed via effective resistance calculations.

Payment networks on cryptocurrencies are particularly beneficial for two kinds of transactions.  First, they streamline transactions between parties that transact frequently, especially between parties with a direct connection.  And second, by increasing transaction throughput, they put downward pressure on transaction fees, which helps make small transactions economically worthwhile.  Again relative to a uniform probability measure (i.e. transaction distributions that are symmetric), our analytical toolkit enables a proof of the monotonicity conjecture for transactions of these types.  

In real world terms, this means that bad actors cannot add edges to the graph and harm either of these payment network use cases.  

\item

Finally, in section 6, we also show that the problem of sampling a random network state reduces to the problem of sampling a random point from a zonotope (that depends on the payment network graph).  Sampling from a distribution on a convex set is a well studied problem and can be done in practice with reasonably fast Markov chains, such as the well-studied Hit and Run process \cite{smith1984efficient}.   

Furthermore, relative to a uniform measure on network configurations, we give an algorithm (Theorem \ref{thm:uniform_sample_algorithm}) for exact uniformly random sampling states of a payment network of $n$ agents and $m$ edges in time $O(m\beta(m))$, where $\beta(m)$ is the complexity of solving a linear program on $O(m)$ constraints.  For comparison, Hit and Run would require $O(n^3\beta(m))$.  Currently, $\beta(m)=m^\omega$ \cite{cohen2019solving}, for $\omega$ the matrix multiplication constant.

\end{enumerate}



\subsection{Related Work}

The concept of the Credit Network that underlies the analysis here was developed independently by \cite{ghosh2007mechanism,defigueiredo2005trustdavis,karlan2009trust}.  These authors worked in the contexts of informal borrowing networks \cite{karlan2009trust}, online reputation systems backed by social relationships \cite{defigueiredo2005trustdavis}, and auction systems built on a credit network.   

Recently, an analysis of the distribution of network states induced by a distribution on transactions was conducted in \cite{branzei2017charge}.  The stated goal of these authors is to understand transaction pricing models, but central to the pricing analysis is an understanding of transaction failure rates in the short term.  By constrast, for scaling a network, the key metric is failure rate amortized in the long term. Their analysis, however, only looks at graphs consisting of isolated edges and of stars, and also assumes that the size of a transaction between pairs is fixed. 

Most early implementations of payment networks discuss transactions as being routed along single paths (as in e.g. \cite{poon2016bitcoin}).  But for the purpose of actually sending money, one successful \$2 transaction along a single path is functionally equivalent to two separate successful \$1 transactions along two different paths.  More generally, a transaction in a payment network can be thought of as a flow through the network.

What we study is also similar to the study of ``capacitated transportation networks,'' as in \cite{doulliez1972transportation} and \cite{hassin1988probabilistic}.  In this model, the capacity of each edge is drawn independently from some random distribution.  These capacities are ``successful'' if they can satisfy some fixed set of flow sources/demands, and the reliability of the network is the chance of success.  Similarly, Grimmett and Suen \cite{grimmett1982maximal} studied the expected size of a maximum flow through a graph, if edge capacities are random.  Efforts to approximately compute transportation network reliability include \cite{lin1995reliability}.

A key difference from our model, however, stems from the fact that routing flow along a cycle does not change the ability of a payment network to satisfy a flow demand.  The set of network states, therefore, is not the product of sets of states on individual edges (unless the network is acyclic).

\section{The Network Model}

In practice, a payment network consists of a set of agents and a set of edges between agents, where each edge has some associated amount of money in escrow.  If an edge runs between two agents $A$ and $B$, then $A$ and $B$ privately track how much of the escrow on the edge belongs to $A$ and how much belongs to $B$.  For example, Figure \ref{fig:ln} gives a simple payment network between a few nodes.


Suppose that agent $A$ wishes to pay one unit of money to agent $B$.  Then the amount of money on the edge between $A$ and $B$ remains constant, but $A$ and $B$ privately decide that of this escrow, $A$'s ownership stake decreases by one, and $B$'s ownership increases by one.  Of course, this transaction would fail if and only if $A$'s ownership stake were less than one.

Formally, then, we can model a payment network as the following.  As transactions occur, the unchanging components of the system consist of a graph $G=(V,E)$, representing nodes and edges, along with a map $c:E\rightarrow\mathbb{R}_{\geq 0}$ that records the capacity of each edge.  An escrow configuration of the network, which changes as transactions happen, will be a map $w:(V\times V)\rightarrow\mathbb{R}_{\geq 0}$ recording, for each edge, how much of the escrow each agent owns.  Specifically, for some edge $(u,v)$, if $w(u,v)=k$, then $u$ owns $k$ units of escrow on edge $(u,v)$, and $w(u,v)+w(v,u)=c(u,v)=c(v,u)$.  

Before proceeding, we would like to note that an escrow configuration $w$ in the above sense is {\it not} the same as the network states discussed in the introduction.  For the purpose of sending money, many escrow configurations can be functionally equivalent, and we will treat equivalent escrow configurations as a single network state.

A simple transaction along a single edge will consist of a sender $u$, a receiver $v$, and an amount $k$, and we will say that the transaction succeeds if and only if $w(u,v)\geq k$.  If this transaction succeeds, the network will be left in the state where $w(u,v)$ is decreased by $k$ and $w(v,u)$ is increased by $k$.

To complete the model, a transaction $\tau$ consists of a sender $x$, a receiver $y$, and a transaction amount $k$, and we will say that $\tau$ succeeds in a configuration $w$ if and only if there exists a commodity flow $f=\lbrace f_(u,v):(u,v)\in E\rbrace$ taking $k$ units of flow from $x$ to $y$ with $w(u,v)\geq f_{(u,v)}$.  Such a transaction change an escrow configuration in the same way that simple transactions of size $f_{(u,v)}$ along each edge $(u,v)$ collectively change an escrow configuration.  

\subsection{Definitions and Basic Properties}

First, note that if agents can successfully perform the same sequences of transactions starting from two different escrow configurations, then in terms of moving money, the configurations are identical.  This motivates the following definition.  

\begin{definition}[Transaction Equivalence]

Two configurations $w_1$ and $w_2$ are {\it transaction equivalent} if, given any sequence of transactions $\tau$, $\tau$ succeeds starting from $w_1$ if and only if $\tau$ succeeds starting from $w_2$.
\end{definition}

Unfortunately, this definition provides no obvious aid in understanding the space of configurations.  For that, consider as an example a cycle graph where each edge has capacity $1$, and pick one direction around the cycle to be ``clockwise''.  Let $w_1$ be the configuration where $w_1(u,v)=1$ if and only if $(u,v)$ points in the clockwise direction, and let $w_2$ be the configuration where $w_2(u,v)=1$ if and only if $(u,v)$ points in the counterclockwise direction.  Then for any two vertices $x$ and $y$,
$x$ could send at most one unit of payment to $y$ along a counterclockwise path in configuration $w_1$ and could send at most one unit of payment along a clockwise path in configuration $w_2$.  In this sense, then, $w_1$ and $w_2$ are transaction-equivalent.

Observe that one could move from $w_1$ to $w_2$ by routing a payment from a vertex to itself in a counterclockwise direction.  This motivates the following definition.

\begin{definition}[Cycle Equivalence (Definition 2, \cite{dandekar2011liquidity})]
Two configurations are {\it cycle-equivalent} if and only if one is reachable from the other by routing a series of payments along cycles.
\end{definition}

We will see later that these two definitions are equivalent (Corollary \ref{thm:txiscycle}), and consequently, that the choice of route when making a transaction does not influence the resulting cycle-equivalence class.


Finally, the ``liquidity'' of the network is the probability that a transaction will succeed from a ``random'' cycle-equivalence class.  We make this definition formal in \ref{sec:markov1}, but intuitively, the distribution of interest is the one that is induced by real-world random processes.  We model this as the stationary distribution of a Markov chain that sequentially attempting to execute randomly generated transactions; this Markov chain has a dimple description if it's state space is thought of as cycle-equivalence classes, not as individual configurations.

\section{From Networks to Continuous Representations}
\label{sec:configurations}
For the rest of the discussion, we will assume that a graph $G=(V, E)$ is connected, and let $n=\vert V\vert -1$ be the size of a spanning tree in $G$.

At a high level, in this section, we will construct a map $S$ from the set of escrow configurations to a convex subset $Z$ of $\mathbb{R}^n$, and show that sending money in the credit network is equivalent to moving in $Z$.  That is to say, every transaction $\tau$ corresponds to some vector $v_\tau\in\mathbb{R}^n$, and if executing $\tau$ from a configuration $w$ results in the configuration $w^\prime$, then $S(w)-v_\tau=S(w^\prime)$.  Furthermore, the preimage of any point in $Z$ is exactly one cycle-equivalence class of configurations.  


This map relies on a representation of the graphic matroid in $\mathbb{R}^n$, subject to the additional constraint that the signed summation of the edges visited while walking in a cycle is $0$.  This condition lets us discuss the ``direction'' of money moving from a vertex $x$ to a vertex $y$ in a well-defined manner.

Fix some orientation $s:(V\times V)\rightarrow \pm 1$ on every edge.

\begin{definition}
Let $D$ (for ``direction'') be a map from $E$ to $\mathbb{R}^n$.  

The {\em direction} of a path $p=(v_0,...,v_k)$ is simply the (signed) summation of $D$ along the path:
\begin{equation*}
\sum_{i=0}^{k-1}s(v_i, v_{i+1})D(v_i, v_{i+1})
\end{equation*}  

The direction of a flow $f$ is similarly defined to be 
$\sum_{(u,v)\in V\times V}f_{(u,v)} s(u,v)D(u,v)$

We will refer to the direction of a path $p$ as $D(p)$ and the direction of a flow $f$ as $D(f)$.

\end{definition}

Sending money in a credit network is akin to sending flow through the underlying residual graph.  We want the direction of the flow to capture something about how a credit network configuration changes as money is sent.  For that, it suffices for a direction map to satisfy the following conditions:

\begin{definition}
 A direction map $D$ is a {\em spanning representation} of the credit network if:

 \begin{enumerate}
 	\item The direction (under $D$) of any path is $0$ if and only if $p$ is a cycle.
 	\item For any spanning tree $T$, the set of vectors $\lbrace D(t)\rbrace_{t\in T}$ spans $\mathbb{R}^n$.
 \end{enumerate}

\end{definition}

\begin{theorem}
\label{thm:goodreps}
Spanning representations exist.
\end{theorem}

The proof is in the appendix, but informally, take the edges of a spanning tree, represent them as the standard basis, and represent the other edges so that cycles sum to $0$.  We will not mention the representation when the choice of representation is clear from context.  The utility of a spanning representation is that it naturally induces a map from configurations to $\mathbb{R}^n$ that has several useful properties.

\begin{definition}

The {\it represented score} of a configuration $w$ with respect to a spanning representation $D$ is $S_D(w)=\sum_{(u,v)\in V\times V:s(u,v)=1}w(u,v)D(u,v)$.\footnote{The analysis would work out the same if we simply summed over all directed edges.  This definition makes later theorem statements cleaner.}  

\end{definition}

The condition that $s(u,v)=1$ means that the summation above is over edges pointed in their positive direction.  

This should be thought of as a generalization of the ``score vectors'' as in e.g. \cite{kleitman1981forests}.  In fact, the represented score map exactly captures the transactional relationships between cycle-equivalence classes.

\begin{theorem}

~
\label{thm:scorecycle}

\begin{enumerate}
	\item[1] Two cycle-equivalent configurations have the same represented score.
	\item[2] If a transaction along path $p$ of size $k$ takes configuration $w_1$ to configuration $w_2$, then $S_D(w_1) -k*D(p)=S_D(w_2)$.
	\item[3] If two paths $p$ and $q$ both go from vertex $x$ to vertex $y$, then $D(p)=D(q)$.
\end{enumerate}
\end{theorem}

A transaction of size $k$ from $x$ to $y$ is a special case of a flow demand, and as a corollary of the above properties, any flow that satisfies a given demand will have the same direction.  Given some spanning representation, then, it is reasonable to think of transactions not just as paths from $x$ to $y$ but also as flow demands or as vectors in $\mathbb{R}^n$.  In particular, we make the following convention:

\begin{definition}
A transaction $\tau$ sending $k$ units of flow from $x$ to $y$ corresponds to a vector $v\in\mathbb{R}^n$ if, for some path from $x$ to $y$, $v=kD(p)$.
\end{definition}

In other words, if executing $\tau$ sends configuration $w_1$ to $w_2$, then $S_D(w_1)-v=S_D(w_2)$.

These definitions are sufficient to analyze the space of cycle-equivalence classes. 

First, recall that the {\it Minkowski Sum} of vectors $v_1,...,v_k$ is the set $\lbrace x : x=\sum_i\lambda_iv_i, \lambda_i\in [0,1]\rbrace$.

\begin{theorem}

\label{thm:cyclemksum}
The image of the set of cycle-equivalence classes, under a spanning representation, is exactly the Minkowski sum of the vectors $\lbrace v_{(x,y)}=D(x,y)*c(x,y)) : (x,y)\in V\times V, s(x,y)=1\rbrace$.
\end{theorem}

As before, the $s(x,y)=1$ requirement simply means that the summation is over all edges, viewed in their positive direction.

With these definitions, we can finally make the following useful characterization of transaction-equivalence.

\begin{corollary}
\label{thm:txiscycle}
Transaction-equivalence is equivalent to cycle-equivalence.

\end{corollary}

\section{Induced Measures on the Configuration Space}
\label{sec:markov1}

As outlined in the introduction, we would like to understand success probability with respect to a ``random'' configuration drawn from the distribution induced by real-world transaction activity.  For this work, we will think of transactions as arriving from some external process; most transactional contexts like online shopping are driven by forces entirely external to the payment network, such as demand for consumer goods.

Our analysis is one of an idealized payment network, so we will model transactions as happening one at a time, communication links never fail, and transaction fees are zero.\footnote{The question of how to optimally set transaction fees is a fascinating open problem for future study.}  In particular, we will model transactions as occuring in sequence, drawn from some distribution.  At every timestep, we draw a new transaction by drawing a pair of vertices who will transact and then a transaction size.  If the transaction is feasible in the current network configuration, we execute the transaction and move to the resulting state.  Otherwise, we disregard the transaction.

\begin{definition} [Transaction Model]
A transaction model for a set of vertices $V$ consists of the following:

\begin{enumerate}
\item The set of pairs of vertices that transact, $X\subset V\times V$.
\item A distribution $\phi$ on $X$.  
\item For each pair $(a,b)\in X$, a continuous probability distribution $k_{a,b}(\cdot)$ on the size of a transaction between $a$ and $b$.
\end{enumerate}

This transaction model generates a transaction by using $\phi$ to choose a sender $a$ and a receiver $b$, and then using $k_{a,b}(\cdot)$ to choose a transaction size.
\end{definition}

For technical convenience, we will allow $k_{a,b}(\cdot)$ to output negative values, and sending negative money from $a$ to $b$ will be considered sending money from $b$ to $a$.  It is without loss of generality, then, to assume that only one of $(a,b)$ or $(b,a)$ is in $X$.

We would like to use a transaction model to induce a distribution on the set of cycle-equivalence classes.   In the discrete credit network case of \cite{dandekar2011liquidity,goel2015connectivity}, the set of cycle-equivalence classes was finite, and so those authors could simply sum over all classes.  However, in our continuous setting, the set of cycle-equivalence classes has infinite cardinality.  We would like to analyze the probability mass of a subset of configurations, but it is not clear how to integrate over a set of cycle-equivalence classes without building a lot of measure-theoretic infrastructure.  Instead, given a fixed spanning representation, it will be much easier to simply map the set of cycle-equivalence classes into $\mathbb{R}^n$ and analyze subsets of $\mathbb{R}^n$.  As a notational shorthand, we will say that the measure of a set of cycle-equivalence classes is simply the measure of the image of the set in $\mathbb{R}^n$ (if said image is measurable).

Given this technical difficulty, when we define the following Markov Chain, we will use as its state space the image in $\mathbb{R}^n$ of the set of all cycle-equivalence classes under a fixed spanning representation.

\begin{definition} [Markov Chain 1]
  Let $G=(V, E, c)$ be a payment network, let $D$ be a spanning representation, let $Z\subset\mathbb{R}^n$ be the image under $D$ of the set of cycle-equivalence classes of $G$, and let $\Xi$ be some transaction model.  The induced Markov chain will be defined as follows:

  Start at some point $z_0\in Z$.

  At each timestep, draw a transaction between vertices $a$ and $b$ of size $k$ from $\Xi$.  If the transaction is feasible from the current state $z_t$, then perform the transaction and update the state of the network accordingly ($z_{t+1}=z_t-kD(a,b)$.  Otherwise, remain in the same state ($z_{t+1}=z_t$).

  The induced distribution on $Z$ is the stationary measure of this Markov chain (if a stationary measure exists and is unique).

\end{definition}

\begin{theorem}
\label{thm:mc1invariant}
The above Markov Chain has an invariant probability measure.

\end{theorem}

We are primarily interested in the case where a transaction model induces a unique stationary measure.  The following condition is sufficient but not necessary; there are many similar sufficient conditions.

\begin{theorem}
\label{thm:mc1uniqueinvariant}
Let $G=(V,E)$ be a graph, and let $\Xi=(X, \phi, \lbrace k_{a,b}\rbrace_{(a,b)\in X})$ be a transaction model on $V$.

Suppose that for each pair $(a,b)\in X$, the support of $k_{a,b}$ contains the interval $[-\varepsilon, \varepsilon]$ for some $\varepsilon>0$.  Then Markov Chain 1 has a unique stationary distribution if the graph $(V, X)$ is connected.

\end{theorem}

Motivated by Theorem \ref{thm:mc1uniqueinvariant}, we call a transaction model that satisfies these conditions a {\it connected} transaction model.

In the next section, we will primarily study the case where the unique stationary measure is uniform.

\begin{theorem}
\label{thm:symmetric}

Let $\Xi=(X, \phi, \lbrace k_{a,b}\rbrace_{(a,b)\in X})$ be a connected transaction model.

If each $k_{a,b}$ is symmetric about $0$, then the induced stationary measure is uniform.

\end{theorem}
Motivated by Theorem \ref{thm:symmetric}, we call a transaction that satisfies these conditions a {\it symmetric, connected} transaction model.

When the stationary measure is unique, the natural definition of transaction success probability is well-defined.

\begin{definition}
Given a spanning representation $D$, let $Z$ be the set of configurations represented in $\mathbb{R}^n$ and let $\pi$ be some density measure over $Z$.  Then the {\it liquidity} of the network with regard to a transaction $v$ is the relative measure of the image of the set of configurations where $v$ succeeds, $\frac{\int_Z \pi(z)\mathds{1}\lbrace z-v\in Z\rbrace dz}{\int_Z\pi(z)dz}$.

\end{definition}

\begin{theorem}
  
\label{thm:reliability}

The liquidity of a network with respect to a particular transaction does not depend on the choice of spanning representation. 

\end{theorem}

These theorems show that under small assumptions on transaction distributions, the distribution on configurations of a payment network is a well-defined topic of discussion.

\section{Liquidity Analysis and Monotonicity}
\label{sec:monotonicity}

One question that an agent operating in a payment network might ask is how adding escrow to one edge will affect transaction success rate across another edge.  It would be bad for a payment network if bad actors could add escrow and in doing so decrease liquidity between well-behaved agents.  

The authors of \cite{goel2015connectivity} were concerned with this question, which they dubbed the ``Monotonicity Conjecture.''  In their work, they observe that this conjecture is equivalent to the well-studied ``Negative Correlation'' conjecture about the set of forests of a graph.  To summarize, a matroid on a set of elements $X$ is said to be negatively correlated if, for any two elements $x,y\in X$, the probability that $x$ is in a random basis is greater than the probability that $x$ is in a random basis that also contains $y$.  \cite{as1992balanced} showed that this conjecture implies a polynomial time random sampling algorithm for matroid bases.  

Not all matroids are negatively correlated, however.  And although the set of forests is not a matroid, much intellectual energy has been devoted to the negative correlation of forests question -- for example, \cite{grimmett2004negative,cocks2008correlated,semple2008negative}.  

We show in this section that for the monotonicity conjecture on credit networks, relaxing the (somewhat unnatural) restriction that transactions must have discrete sizes can actually make the monotonicity question tractable.  Below we reproduce this conjecture in the context of our network model.

\begin{conjecture}[Monotonicity]

Let $G=(V,E)$ be a credit network with edge capacities $c(\cdot)$, let $\tau$ be a transaction from a vertex $x$ to another vertex $y$ of size $k$, and let $\pi$ be the uniform stationary measure induced by a symmetric, connected transaction model.

For any pair of vertices $(w,z)\in V\times V$, increasing $c(w,z)$ in $G$ should not decrease the liquidity of $\tau$ with respect to $\pi$.

\end{conjecture}

\subsection{Liquidity Analysis}

To move towards a study of the monotonicity conjecture, we first prove some theorems about general liquidity in credit networks in the continuous model.

As before, to avoid the need to set up measure-theoretic infrastructure on the set of cycle-equivalence classes, we will for the rest of this section measure the probability of a set of configurations as the relative measure of its image in $\mathbb{R}^n$ under the map induced by a spanning representation.  When the measure $\pi$ is uniform, the measure of a set in $\mathbb{R}^n$ is simply its volume.  

Therefore, as a notational shorthand, we take the ``volume'' of a set of cycle-equivalence classes to mean the volume in $\mathbb{R}^n$ of the image of the set of cycle-equivalence classes under the map induced by a fixed spanning representation.  As per Theorem \ref{thm:reliability}, liquidity metrics are not affected by choice of spanning representation.

As mentioned in Theorem \ref{thm:cyclemksum}, the image of the set of configurations is the Minkowski sum of a particular set of vectors.  It turns out that the volume of this Minkowski sum is closely related to the undirected graphical matroid of the payment network.  The connection comes through the generating polynomial for the matroid, which we now define for completeness.

\begin{definition}

  Let $G(V,E)$ be a graph, and let $T(G)$ be the set of all spanning trees of $G$.  Associate with every $e\in E$ a real variable $y_e$.

  The {\it generating polynomial} of $T(G)$ is $\Gamma(y)=\Sigma_{T\in T(G)} \Pi_{e\in T}y_e$.

  We will also denote as $\Gamma_e(y)$ the summation over sets $S$ such that $S\cup \lbrace e\rbrace$ is a spanning tree, or equivalently, $\frac{\partial}{\partial y_e}\Gamma(y)$.

  More generally, we will denote with $\Gamma_{a=b}(y)$ the generating polynomial of the graph where vertex $a$ is identified with vertex $b$.  If $e=(a,b)$, then $\Gamma_{a=b}=\Gamma_e$.
\end{definition}

\begin{theorem}
\label{thm:zonovolume}
  Let $G=(V, E)$ be a payment network with edge capacities $c(\cdot)$.

  Suppose that a spanning representation $\hat{D}$ maps some tree to the standard basis.
  Then the volume of the image of the set of cycle-equivalence classes of configurations is equal to $\Gamma (c(e_1),...,c(e_m))$.
\end{theorem}

The volume of a Minkowski sum of vectors $v_1,...,v_m$ is equal to a summation over subset of vectors $v_{i_1},...,v_{i_n}$ of the determinant of the matrix whose columns are $v_{i_1},...,v_{i_n}$.  The result follows from the fact that an edge can only appear once in a simple cycle and the multilinearity of the determinant.

For convenience, in the rest of this section, assume that we have fixed $\hat{D}$ to be such a spanning representation.  \footnote{As per Theorem \ref{thm:reliability}, choice of spanning representation does not affect liquidity results.  This makes the presentation of results simpler.}

To analyze liquidity, we also need a useful analytical description of the volume of set where a transaction $v$ succeeds.

\begin{lemma}
  Let $v$ be the vector in $\mathbb{R}^n$ corresponding to some transaction $\tau$
  If $Z$ is the subset of $\mathbb{R}^n$ corresponding to the entire set of cycle-equivalence classes, then $Z\cap \lbrace z:z-v\in Z\rbrace$ (Minkowski subtraction) is the space where $\tau$ succeeds.
\end{lemma}
\begin{proof}
  Follows from the theorems of Section \ref{sec:configurations}. 
\end{proof}

We start with the easy case for analysis.  Theorems \ref{thm:cyclemksum} and \ref{thm:zonovolume} shows that the image $Z$ of the set of all cycle-equivalence classes is a zonotope with a convenient structure.  When there is an edge already present in the graph between two agents $x$ and $y$, the area where a transaction from $x$ to $y$ succeeds retains the same kind of structure.

In the language of the above lemma, for a transaction $v$ from $x$ to $y$, the space corresponding to transactions where $v$ succeeds, $Z\cap \lbrace z:z-v\in Z\rbrace$, is a zonotope.  In fact, this zonotope is identical to the zonotope generated by a credit network where the capacity of the edge from $x$ to $y$ is reduced by the size of the transaction.

\begin{theorem}
  \label{thm:mkdiffgood}
  If $v$ corresponds to a transaction between $x$ and $y$ of size $k$, and $c(x,y)\geq k$, then the volume where $v$ succeeds is equal to $\Gamma(c^\prime(e_1)...,c^\prime(e_m))$ where $c^\prime=c$ except that $c^\prime(x,y)=c(x,y)-k$.  The volume of the space where $v$ fails is equal to $k\Gamma_{x=y}(c(e_1)...,c(e_m))$.
\end{theorem}

When this holds, then computing a transaction's chance of success is straightforward.

\begin{corollary}
  \label{thm:volume_resistance}
  Suppose that $v$ corresponds to a transaction between $x$ and $y$ of size $k$, and that the graph contains an edge $(x,y)$ with $c(x,y)\geq k$.

  \begin{enumerate}
  \item[1] The probability that $v$ fails is simply $k \Gamma_{(x,y)}(c(e_1),...,c(e_m))/\Gamma(c(e_1),...,c(e_m))$.
  \item[2] Let $G=(V,E)$ be the underlying graph, and let $G^\prime=(V, E^\prime)$ be a multigraph where $E^\prime$ is $E$ with $(x,y)$ duplicated to make edges $f_1$ and $f_2$.  Set $c^\prime(f_1)=k$, $c^\prime(f_2)=c(x,y)-k$, and for all other edges $e$, $c^\prime(e)=c(e)$.  

  Then the probability that $v$ fails is equal to the probability that $f_1$ is in a random weighted tree sample from $G^\prime$.
  \end{enumerate}
\end{corollary}

As an aside, the criterion that $x$ and $y$ have an edge of capacity at least $k$ does not mean that $x$ and $y$ must route all of their transactions explicitly along that edge.  The model does not require any particular type of routing algorithm.  The rest of the graph is still relevant for $x$ and $y$ when they share an edge, as many configurations will require $x$ and $y$ to use edges other than the one directly between them.  

Now for the harder case for analysis.  Namely, when $n>2$ and $v$ does not correspond to an edge in the underlying graph, the space where $v$ succeeds is not necessarily a Minkowski sum of a set of vectors and thus lacks a concise mathematical characterization.  Such examples unfortunately occur in very simple graphs, such as the cycle on four vertices when nonadjacent vertices transact.  Figure 6 of \cite{althoff2015computing} gives a visual depiction of such an object.

Nevertheless, lower bounds on liquidity are derivable using the same analytical frameworks as Corollary \ref{thm:volume_resistance}.

\begin{lemma}
  \label{thm:smalltx}
  Let $x$ and $y$ be any two vertices, and let $v$ be the direction of a path from $x$ to $y$. 

  Let $f_v(k)$ denote the probability that a transaction size $k$ from $x$ to $y$ fails.  Then $0\leq f_v^\prime(k)\leq \Gamma_{x=y}(c)/\Gamma(c)$.  
  Moreover, $f_v^\prime(0)=\Gamma_{x=y}(c)/\Gamma(c)$.  
  
\end{lemma}

\begin{corollary}

  A lower bound of transaction success probability (i.e. $1-k*f_v^\prime(0)$) is computable efficiently (via calculating effective resistance, similarly to Theorem \ref{thm:volume_resistance}).

\end{corollary}

\begin{proof}
  Let $Z$ be the image of the set of all cycle-equivalence classes, let $P$ be the projection of $Z$ onto some fixed hyperplane with normal $v$, and let $\rho$ denote the volume of the projection in $\mathbb{R}^{n-1}$.  For every $p\in P$, consider the line extending in direction $v$ from $p$.  Suppose a point $q$ is on this line.  Let $d_p(q)$ be the signed distance of $q$ from $p$, namely, $\vert p-q\vert\cdot sgn(v\cdot (q-p))$.

  Then for the set of points $Z_p$ on this line and also in $Z$, let $g_1(p)=\min_{z\in Z_p}(d_p(z))$ and $g_2(p)=\max_{z\in Z_p}(d_p(z))$.

   Note that because $Z$ is convex, it can be represented as the space $\cup_{p\in P}\lbrace x : x\in p+[g_1(p), g_2(p)]\rbrace$.

   Then the volume of $Z$ where a transaction in the direction of $v$ of size $k$ fails is:
   \begin{equation*}
  \int_0^k\int_P\mathds{1}(g_2(p)-g_1(p)\geq s)dp~ds.
  \end{equation*}
 The derivative of this quantity at $k$ is therefore 
  \begin{equation*}
  \int_P\mathds{1}(g_2(p)-g_1(p)\geq k)dp
  \end{equation*}

  Clearly this quantity is bounded between $\rho$ and $0$, achieves $\rho$ when $s=0$, and is monotonically decreasing in $k$.

  Note that a projection of a Minkowski sum of a set of vectors is the Minkowski sum of the vectors individually projected.  It follows that $\rho =\Gamma_{x=y}(c)$ and that $f^\prime_v(0)$ is $\Gamma_{x=y}(c)/\Gamma(c)$. 

  $1-kf_v^\prime(0)$ lower bounds transaction success probability because the volume where a transaction of size $k$ in direction $v$ fails is upper bounded by $k\Gamma_{x=y}(c)$, and thus transaction success probability is at least $\left(\Gamma(c)-k\Gamma_{x=y}(c)\right)/\Gamma(c)=1-kf_v^\prime(0)$.

\end{proof}

In fact, the Lemma \ref{thm:smalltx} directly generalizes \ref{thm:volume_resistance}.  

\begin{corollary}
\label{cor:tight_cond}
The lower bound of Lemma \ref{thm:smalltx} is tight if and only if there is an edge of capacity at least $k$ in the direction of the transaction in the credit network graph.

\end{corollary}

\begin{proof}
The lower bound is clearly tight if the credit network graph contains an edge in the direction of the transaction, by Theorem \ref{thm:volume_resistance}.

Now suppose that the edge does not exist.  Specifically, let $Z$ again be the image of the set of all cycle equivalence classes, and let $v$ be the direction of a transaction between vertices $x$ and $y$.  Suppose that the transaction has size $k$, and that edge from $x$ to $y$, if it exists, has capacity $c(x,y)<k$.  

The proof of Lemma \ref{thm:smalltx} shows that volume of the space in $Z$ where a transaction in direction $v$ of size $k$ fails is equal to 
 \begin{equation*}
  \int_0^k\int_P\mathds{1}(g_2(p)-g_1(p)\geq s)dp~ds.
  \end{equation*}
  with $P$, $g_1$, and $g_2$ defined as in the proof of Lemma \ref{thm:smalltx}.

  The lower bound of Lemma \ref{thm:smalltx} is tight if and only if 
  \begin{dmath}
  \label{eqn:tight_if_contained}
  \int_0^k\int_P\mathds{1}\left(g_2(p)-g_1(p)\geq s\right)dp~ds = k\int_P\mathds{1}{\left(g_2(p)-g_1(p)\geq 0\right)dp}=k\Gamma_{x=y}(c)
  \end{dmath}

It suffices to find a point $z\in Z$ such that $z+(k/2)v\notin Z$ and $z-(k/2)v\notin Z$.  Given such a point $z$, the point $z+(k/2)v$ is a point from which $z+(k/2)v-kv=z-(k/2)v\notin Z$.  Because $Z$ is closed, there is a point $p_0\in P$ such that $g_2(p_0)-g_1(p_0)<k-\delta$.  Hence there exists $\varepsilon$ such that for $p\in P$ with $\vert\vert p-p_0\vert\vert_2\leq \varepsilon$, then $g_2(p)-g_1(p)<k-\delta/2$, and thus the first equality of Equation \ref{eqn:tight_if_contained} cannot hold.

Consider a credit network configuration defined as follows.  Pick some vertex cut $H\subset V$ of the graph separating $x$ from $y$.  That is, a set of vertices such that all paths from $x$ to $y$ flow through $H$ (other than possibly a direct path from $x$ to $y$).  Let $m$ be the number of edges $(h, i)$ such that $h\in H$ and $i\notin H$, and for every such edge $(h, i)$, set $w(i, h)$ such that $w(i,h)\leq (k-c(x,y))/(2m+1)$.  Let $w(x,y)=w(y,x)=c(x,y)/2$.  Pick $w$ for every other edge arbitrarily.

Then from configuration $w$, strictly less than $c(x,y)/2+m(k-c(x,y))/(2m)=k/2$ money can move from $x$ to $y$ or from $y$ to $x$.  This means that the point $z\in Z$ that is the image of $w$ under our spanning representation has the property that $z+kv/2\notin Z$ and $z-kv/2\notin Z$ but $z\in Z$.

\end{proof}

\subsection{Monotonicity}

Although we conjecture monotonicity in the general case, the same technical difficulty (Minkowski subtraction does not always produce zonotopes) that obstructed liquidity analysis presents itself when studying monotonicity.  
Nevertheless, we show that monotonicity holds when transactions are small or between vertices that already share a direct connection in the credit network graph. 

\begin{theorem}

  Suppose that $\tau$ is a transaction from $x$ to $y$ of size $k$ and that $c(x,y)\geq k$.
  Let $a$ and $b$ be any two vertices (not necessarily connected directly, possibly equal to $x$ or $y$).

  Then increasing $c(a,b)$ does not decrease the success probability of $\tau$.

\end{theorem}

\begin{proof}
  If $a$ and $b$ are disconnected, connecting $a$ to $b$ with an edge of capacity $0$ does not change the liquidity of any transaction, so without loss of generality assume $a$ and $b$ are connected by an edge of possibly $0$ capacity.

  Let $Z$ be the image of the entire space of configurations before adding capacity on edge $(a,b)$ , and let $A\subset Z$ be the image of the space where $\tau$ succeeds.  Let $Y$ be the image of the entire space of configurations after capacity is increased along $(a,b)$, and let $B\subset Y$ be the image of the space where $\tau$ succeeds after capacity is increased along $(a,b)$.
 
  For convenience, write $c=(c(e_1),...,c(e_m))$ and $c^\prime=(c(e_1),..., c(x,y)-k,...,c(e_m))$.   Let $\mu(\cdot)$ denote the volume of a set in $\mathbb{R}^n$. Finally, let $h$ be the increase in capacity along $(a,b)$.

  It suffices to show that $\mu(B)/\mu(Y) \geq \mu(A)/\mu(Z)$.

  Note that $Z$ is the Minkowski sum of the vectors $\lbrace c(e)D(e)\rbrace_{e\in E}$ and that $A$ is the Minkowski sum of the vectors $\lbrace c^\prime(e)D(e)\rbrace_{e\in E}$.

  Then $\mu(Z)=\Gamma(c)$ and $\mu(A)=\Gamma(c^\prime)=\Gamma(c)-k\Gamma_{x=y}(c)$.  

  Moreover, adding $w$ to $A$ and $Z$ to get $B$ and $Y$ means that $B$ and $Y$, respectively, are the Minkowski sums of $A$ and $Z$ with the the vector $hw$.  Hence, $\mu(Y) = \mu(Z) + h\Gamma_{a=b}(c)$ and $\mu(B) = \mu(A)+ h\Gamma_{a=b}(c^\prime)=\Gamma(c)-k\Gamma_{x=y}(c)+h\Gamma_{a=b}(c)-kh\Gamma_{x=y,a=b}(c)$.  

  Rearranging terms shows that inequality holds if and only if $\Gamma_{x=y, a=b}(c)\Gamma(c)\leq \Gamma_{a=b}(c)\Gamma_{x=y}(c)$, which holds because the tree matroid is Rayleigh \cite{choe2006rayleigh,kirchhoff1847ueber}. 

\end{proof}

More generally, we can prove the monotonicity conjecture if we bound transaction size.

\begin{theorem}
For every credit network $G$, there exists a constant $\hat{k}$ such that $G$ is monotone with regard to transactions of size less than $\hat{k}$.
  
\end{theorem}

\begin{proof}

  Let $v$ be the direction of a path from $x$ to $y$, and let $w$ be the direction of a path from $a$ to $b$.  Let $c=\lbrace c(e_i)\rbrace_{e_i\in E}$.  Finally, let $f_v(k, h)$ denote the probability that a transaction sending flow of size $k$ in direction $v$ fails in the graph where capacity $h$ has been added to an edge in direction $w$.

  First, note that by Lemma \ref{thm:smalltx}, at $k=0$, $\frac{\partial}{\partial k}f_v(k, h)=(\Gamma_{x=y}(c)+h\Gamma_{x=y, a=b}(c))/(\Gamma(c)+h\Gamma_{a=b}(c)$.

  Because the tree matroid is Rayleigh \cite{choe2006rayleigh}, some computation shows that $\frac{\partial}{\partial k}f_v(k, 0) \leq \frac{\partial}{\partial k}f_v(k, h)$ at $k=0$ for every $h>0$.  These two quantities are equal if and only if there are no cycles containing the edges $(a,b)$ and $(x,y)$, in which case monotonicity holds trivially.

  Then for any fixed $h$, there exists some $\overline{k}>0$ such that for all $k<\overline{k}$, $f_v(k, 0)\leq f_v(k, h)$.

  Now, let $A$ denote the volume where the transaction succeeds before the capacity addition, and let $B$ denote the volume of the entire space before the capacity addition.  Let $C_h$ denote the marginal volume where the transaction succeeds after the capacity addition, and let $D_h$ denote the marginal volume of the entire space after the capacity addition.

  Importantly, because the space of configurations is a Minkowski sum, adding capacity simply adds an element to the Minkowski sum.  The volume of any convex subset of the Minkowski sum, therefore, is increased simply by $h$ times its projection in the direction $w$.  Thus the ratio $C_h/D_h$ is constant with respect to $h$.

  Note that $f_v(k,h)=1-(A+C_h)/(B+D_h)$.  Then for any $k$, if $f_v(k,h)\leq f_v(k,0)$ for any $h>0$, then $f_v(k, h)\leq f_v(k, 0)$ for all $h>0$.  

  Putting this all together, there exists some $\hat{k}>0$ such that the credit network is monotone with respect to all transactions along paths with transaction size less than $\hat{k}$.

\end{proof}

\section{Efficient Sampling}

In Theorem \ref{thm:mkdiffgood}, we showed that for a narrow class of parameters, agents can compute the probability of transaction success.  But this class of parameters was quite narrow, and likely not directly relevant to real-world networks.  An agent in a payment network might like to be able to estimate this probability efficiently, so that they could measure the effects of behavior change under a wide variety of distributional assumptions.

In prior work on credit networks that assume that transactions have discrete sizes, estimating this probability is roughly equivalent to sampling a random forest of the underlying graph.  A polynomial-time algorithm for this was recently discovered in \cite{anari2018log}.  This algorithm is quite interesting from a theoretical standpoint, but likely is still too slow to run on real-world graphs.

By contrast, we observe here that removing the discrete size requirement makes sampling a random network state quite efficient for a wide variety of parameter regimes.  Where earlier work on credit networks invoked long lines of detailed mathematical machinery, sampling a random network state here simply requires sampling a uniformly random point from a convex body.  This can be done with the well-studied hit-and-run sampling scheme \cite{smith1984efficient}.
\begin{theorem}

If $\pi$ is a log-concave distribution on the set of cycle-equivalence classes of a credit network, then the hit-and-run sampling algorithm samples a random cycle-equivalence class from $\pi$.

The algorithm requires a preprocessing step of time complexity $\widetilde{O}(n^5)$ preprocessing and (amortized) $\widetilde{O}(n^3\beta(\vert E\vert))$ per sample, where $\beta(x)$ denotes the complexity of solving a linear program with $x$ constraints (hiding log factors and dependencies on $1/\varepsilon$).

\end{theorem}

By the results of \cite{cohen2019solving}, $\beta(x)\leq x^\omega$, for $\omega$ the current matrix multiplication constant.

\begin{proof}
The zonotope is a convex body, and thus this statement is simply a special case of the results of \cite{lovasz2003hit} and \cite{lovasz2007geometry}.  The algorithm, however, requires $\widetilde{O}(n^3)$ zonotope membership oracle calls.  Checking membership in a zonotope can be solved with a linear program on $O(\vert E\vert)$ variables. 

\end{proof}



When the distribution $\pi$ is uniform, we can improve our sample time complexity to the minimum of $O(\vert E\vert\beta(\vert E\vert))$.  The image under a spanning representation of the space of cycle-equivalence classes is not an arbitrary zonotope, but rather, has additional structure that we can exploit.

\begin{theorem}
\label{thm:uniform_sample_algorithm}
If $\pi$ is the uniform distribution on cycle-equivalence classes of a credit network, then there exists an exactly uniform sampling algorithm with time complexity $O(\vert E\vert (\beta(\vert E\vert)+\vert E\vert^\omega))$.

\end{theorem}

This shaves a factor of $\vert V\vert^3/\vert E\vert$ off of the asymptotic sampling complexity.  Note that real-world payment network graphs, like those deployed in Lightning, are typically quite sparse.  

Our algorithm runs in $\vert E\vert$ steps, where each step involves solving a linear program on $\vert E\vert$ variables (taking time $\beta(x)$) and computing the electrical resistance between two points in a graph (taking time $\vert E\vert^\omega$).  We remark that using the linear program solver of \cite{cohen2019solving} gives total complexity of $O(\vert E\vert^{\omega+1})$.  However, one could substitute a faster approximate resistance algorithm and get an approximately uniform sample with a faster complexity.  Similarly, one could use any linear program solver in our algorithm, and the time complexity would change accordingly.

\begin{proof}

Let $G=(V,E)$ be a credit network, and fix some ordering on the edges $e_i$ and some spanning representation $D$.  Furthermore, say that the first $n$ elements of the ordering form a spanning tree.  Let the edge capacity of edge $e_i$ be $c(e_i)$, and let the Minkowski sum of the first $k$ edges be $Z_k$.  Denote $m=\vert E\vert$.  It suffices to sample a uniformly random point within $Z_m$.

Let $e_i$ be any edge.  Let $S_i$ denote the portion of the surface of $Z_i$ that is ``visible'' from the direction of $e_i$.  (The surface of a zonotope is a collection of faces.  A face forms part of $S_i$ if its normal vector $n$ has $n\cdot D(e_i)\geq 0$).  Note that $Z_i$ can be decomposed into the almost disjoint union of two pieces.  The first is $Z_{i-1}$, and the second is the remainder, which is the Minkowski sum of the set $S_{i-1}$ and the vector $D(e_i)c(e_i)$.  Their intersection is a translation of $S_{i}$.

As noted in Corollary \ref{thm:volume_resistance}, the ratio of the volume of the second part to the total volume of $Z_i$ is computable via an effective resistance calculation.

Let $P_i$ be the projection of $S_i$ onto the hyperplane through the origin with normal $D(e_i)$.  A random point in this second part can be sampled by taking a random point in $P_{i-1}$, mapping back to the corresponding point in $S_{i-1}$, and adding $\lambda_ic(e_i)D(e_i)$ for $\lambda_i$ drawn from $(0, 1]$ uniformly at random.  

The mapping from points in $P_i$ back to points in $S_i$ can be computed via a linear program. Specifically, given a point $p$, maximize $\alpha$ such that $p+\alpha D(e_i)$ is a point within $Z_i$.  This linear program has $O(n+m)$ constraints and $m$ decision variables, so solving this $LP$ takes time $O(\beta(m))$.

Conveniently, both $Z_{i-1}$ and $P_{i-1}$ are zonotopes, so our sampling process naturally sets itself up for a recursive sampling algorithm.

An algorithm $\mathcal{A}$ for sampling a random point from the Minkowski sum, therefore, could look as follows.  Given as input vectors $v_1,...,v_k\in\mathbb{R}^n$ and capacities $c(e_1),...,c(e_k)$, $\mathcal{A}$ will return a uniformly random point in the Minkowski sum of the input vectors.  For notation, let this Minkowski sum again be $Z_k$ and its surface in the direction $D(e_k)$ be $S_k$.   
First, check whether the input set of vectors is linearly independent.  If they are, return a random list of $k$ values drawn uniformly at random from $[0,1]$.

Next, compute the effective resistance $r$ along edge $k$.  Use this to choose between the spaces $Z_k\setminus Z_{k-1}$ and $Z_{k-1}$ randomly proportionally to their volumes.

If the algorithm chooses $Z_k\setminus Z_{k-1}$, the algorithm can invoke $\mathcal{A}$ on the list of vectors $(w_1,...,w_{k-1})$, where $w_i$ is the component of $v_i$ perpendicular to $v_k$.  It can then find the corresponding point $s$ in $S_k$ via the above linear program.  

Finally, the algorithm can sample $\lambda_k$ uniformly at random from $(0, 1]$ and return $s+\lambda_kD(e_k)$.

If we do not accept $e_k$, simply return $\mathcal{A}(w_1,...,w_{k-1}, c(e_1),...,c(e_{k-1}))$.

Computing a spanning representation takes time $O(nm)$.  Each step of this algorithm takes time $O(m^\omega+\beta(m) + m)$ for solving the linear program, computing effective resistance, and processing a list of vectors.  Since the algorithm recurses $m$ times, the total time complexity is $O(m*(m^\omega+\beta(m))$. 
\end{proof}

\section{Future Work, and A Note on Symmetricity}

Our arguments for the monotonicity conjecture require that the Markov-chain-induced probability measure on the set of configurations be uniform over cycle-equivalent classes of configurations.  We showed that this occurs when the transaction distribution is symmetric.  However, real-world transaction distributions may not be symmetric.

In the real world, transactions on a payment network do not necessarily happen in sequence.  We chose to analyze a discrete-time transaction model for conceptual simplicity and to follow the pattern of prior work on credit networks; however, we could have analyzed a continuous-time transaction model where in some window of time, pairs of agents can transact many times.  Within a fixed time window, the net transaction balance between every pair of agents $u$ and $v$ might be randomly distributed.  We could then sum these biases, weighted by the direction of a transaction from $u$ to $v$, to get an overall direction of bias $\nu\in\mathbb{R}^n$.  

With this motivation, we could also model transactions via a ``Reflected Brownian Motion'' with drift vector $\nu$ (and normal reflection).  It follows directly from \cite{harrison1987multidimensional} that the invariant distribution $\mu(x)$ of an RBM with bias $\nu$ is proportional to $e^{x\cdot\nu}$.  

Of course, in the symmetric distribution case, $\nu$ is $0$ and the distribution is again uniform.  If $\nu$ is nonzero, then the invariant distribution is efficiently sampleable using an algorithm like Hit and Run.  Some asymmetric transaction distributions still have a net bias of $0$.  Finding realistic real-world assumptions that would be sufficient to imply uniformity of the induced stationary measure, or analyzing monotonicity when the induced stationary measure is nonuniform, are interesting open research problems.  

General models of Brownian motion in a bounded space have been studied in other contexts, such as network queuing theory.  A credit network is not exactly a traditional queuing network, but nevertheless, understanding connections with results in queuing theory or the theory of diffusion processes might prove fruitful.  A potential line of future work would be to study cases where agent behavior changes based on the current configuration of the network.  Under mild technical conditions, the notion of an induced invariant measure when agent behavior varies with state is still well-defined, and analogous concepts have been studied in other literature (such as \cite{kang2014characterization}).

\section{Conclusion}
  The Credit Network model forms an accurate representation of the combinatorics underlying modern ``Layer-2'' transaction protocols.  The performance of these networks is an important piece of the process of scaling up cryptocurrency deployments.

  Prior work was only able to analyze the scenario where transactions have discrete sizes, and could only give meaningful liquidity bounds when transactions had a single fixed size that was not too small relative to edge capacities.  Not only is this not a natural assumption, but the resulting analysis left open the important question of monotonicity and the problem of sampling efficiently random network configurations.  In this work, we show that not only is the analogous continuous model analytically tractable, but that removing the assumption of discrete transaction sizes actually enables efficient sampling and progress on the monotonicity conjecture.  

  The results are obtained by exploiting a relationship between the graphs underlying payment networks and particular convex sets in $\mathbb{R}^n$.  The problems of monotonicity and sampling thereby transform into problems of analyzing the volume of convex sets.  This transformation relies crucially on the combinatorial properties of the underlying graphic matroid.

  The authors hope that connections to other branches of mathematics, such as diffusion processes, that the continuous relaxation facilitates could prove fruitful areas of research, such as when studying state-dependent agent behavior or transaction pricing.

\begin{acks}
This research was supported by the Stanford Center for Blockchain Research and by the Office of Naval Research, award number N000141912268.
\end{acks}

\bibliographystyle{splncs04}
\bibliography{unified-bib}

\appendix
\section{Section \ref{sec:configurations} Omitted Proofs}

\subsection{Proof of Theorem \ref{thm:goodreps}}

\begin{proof}
There are, in fact, many spanning representations.  This is unsurprising, given that graphical matroids are representable over $\mathbb{R}$, and the dimension of any basis of the graphical matroid is $n$.  It suffices to show that there exists a matroid representation such that summing along a cycle results in $0$, not just that the set of vectors corresponding to a cycle forms a linearly dependent set.

As a concrete example, let $T$ be a spanning tree, and let $v_1,...,v_n$ be a basis for $\mathbb{R}^n$.  Correspond the edges in $T$ with $v_1,...,v_n$.  For any other edge $(x,y)$ with $s(x,y)=1$, let $p_0=x,...p_t=y$ be the path from $x$ to $y$ in $T$, and set $D(x,y)=\sum_{i=0}^{t-1}s(p_i, p_{i+1})D(p_i, p_{i+1})$.  

Then clearly the direction walking along the cycle formed by $T$ with $(x,y)$ is $0$.  

Note that the direction summation is clearly additive.  In other words, if a path $p$ runs from $x$ to $y$ and a path $q$ runs from $y$ to $z$, then the direction of $p$ plus the direction of $q$ is equal to the direction of the concatenation of $p$ and $q$.

Observe that walking from $x$ to $y$, and then back along the same path, gives a path whose sum is $0$.  In fact, any path from $x$ to $y$ that visits a vertex $v$ twice can be decomposed into a path directly from $x$ to $y$ plus a path from $v$ to itself.  Sufficiently repeating this decomposition shows that the sum along every path in $T$ from $x$ to $y$ has the same summation.  Note that because $v_1,...,v_n$ form a basis for $\mathbb{R}^n$, this summation is $0$ if and only if $x=y$.  

For any other path $p=(p_0,...,p_t)$, if $(p_i, p_{i+1})\notin T$, then, by the construction, the summation of the path along $p$ is equal to the summation along the path $p^\prime=(p_0,...,p_i, q_1,...q_k, p_{i+1},...,p_t)$, where $(p_i, q_1,...q_t, p_{i+1})$ is a path from $p_i$ to $p_{i+1}$ in $T$.  Repeating this process until all edges are in $T$ shows that $p$ has direction $0$ if and only if $p$ is a cycle.


\end{proof}

\subsection{Proof of Theorem \ref{thm:scorecycle}}

\begin{proof}\hspace*{\fill}

\begin{enumerate}
	\item[1] This follows from part 2.  Assuming part 2, if transactions along cycles take $w_1$ to $w_2$, then, because cycles have direction $0$, $S(w_1)=S(w_1)+0=S(w_2)$.
	\item[2] Suppose that a path $p=(p_0,...,p_t)$ is feasible in configuration $w_1$.

	By rearranging the summations, we get that

  \begin{dmath*}
   S(w_1)-S(w_2)=\Sigma_i s(p_i, p_{i+1})k*D(p_i, p_{i+1})=k*D(p)
  \end{dmath*}
	
  \item[3]   Let $p^\prime$ be the reversal of $p$.  Then $p$ concatenated with $p^\prime$ forms a cycle, and $q$ concatenated with $p^\prime$ forms a cycle, so 

  \begin{dmath*}
  D(p)=-D(p^\prime)=D(q)
 \end{dmath*}
 
\end{enumerate}
\end{proof}

\subsection{Proof of Theorem \ref{thm:cyclemksum}}
\begin{proof}

This theorem follows almost directly from the construction.  If a vector $v$ is expressable as a Minkowski summation with coefficients $\lbrace \lambda_{(x,y)}\rbrace_{(x,y)\in E, s(x,y)=1}$, then the weight map $w$, defined by $w(x,y)=\lambda_{(x,y)}c(x,y)$ if $s(x,y)=1$ and $c(x,y)-w(y,x)$ otherwise, forms a feasible network configuration.  
Conversely, if $S(w)=\sum_{(x,y)\in E:s(x,y)=1}\lambda_{(x,y)}c(x,y)D(x,y)$, then $\lambda_{(x,y)}\in [0,1]$, so $S(w)$ is representable as a Minkowski sum of the vectors of interest.
\end{proof}

\subsection{Proof of Corollary \ref{thm:txiscycle}}
\begin{proof}
  Let $Z\subset\mathbb{R}^n$ be the space of configurations under a spanning representation.  Let $x_1$ and $x_2$ be in $Z$.  Clearly, if $x_1=x_2$, that is, $x_1$ and $x_2$ are cycle-equivalent, then for any transaction $v$, $x_1+v\in Z$ if and only if $x_2+v\in Z$.

  Suppose that $x_1$ and $x_2$ are transaction-equivalent but not cycle-equivalent.  Then the generalized transaction $v=x_2-x_1$ is feasible starting from $x_1$ and takes $x_1$ to $x_2$.  Let $k=\sup\lbrace t:tv+x_2\in Z\rbrace + 1/2$.  Because $Z$ is convex, the transaction $kv\neq 0$ is feasible starting from $x_1$ but not from $x_2$, which gives a contradiction.
\end{proof}

\section{Section \ref{sec:markov1} Omitted Proofs}

\subsection{Proof of Theorem \ref{thm:mc1invariant}}
\begin{proof}

Let $\Xi=(X, \phi, \lbrace k_{a,b}\rbrace_{(a,b)\in X})$ be a transaction model on $V$, and let $D$ be some spanning representation.

Let $p$ be the transition kernel of Markov Chain 1 using transaction distribution $\Xi$, and let $g$ be any continuous, bounded function on the compact state space $Z$.  By Theorem 1.10 of \cite{eberle2009markov}, it suffices to show that if a sequence $x_n$ converges to $x$, then $\int g(y)p(x_n, dy)$ goes to $\int g(y) p(x, dy)$.  But $p$ is the weighted sum of continuous, bounded functions, so this holds.

\end{proof}

\subsection{Proof of Theorem \ref{thm:mc1uniqueinvariant}}
\begin{proof}

Let $\Xi=(X, \phi, \lbrace k_{a,b}\rbrace_{(a,b)\in X})$ be a transaction model on $V$, let $p$ be the transition kernel of Markov Chain 1 using transaction distribution $\Xi$, and let $D$ be some spanning representation.

Let $Z$ be the image under $D$ of the entire space of configurations, and let $\mu$ be some invariant distribution.

Note that it is equivalent to view the set of transacting pairs, $X$, as a set of vectors $X^\prime$.  Moreover, $(V,X)$ is connected if and only if $X^\prime$ spans $Z$.

By Theorem 1 of \cite{harris1956existence}, it suffices to show that if a set $E\subset Z$ has nonzero measure, then from any starting point $z\in Z$, the Markov Chain must visit $E$ infinitely many times in expectation.

Note that it suffices to show that for all $z\in Z$, there is some $k$ such that $p^k(z, E)\neq 0$.  If this holds for all $z$, then in expectation, the chain starting at $z$ will visit $E$ in a finite number of steps.  After this, the chain could either stay in $E$ forever (in which case, the chain trivially visits $E$ infinitely many times), or it will leave $E$ to some $z^\prime$.  

Suppose that $X^\prime$ spans $Z$.

Because each $f_v$ is supported on $[-\varepsilon, \varepsilon]$, for any $z$, the conditional distribution started at $z$ and running for $\vert X\vert$ steps has support on an $L_\infty$ ball (subset of $Z$) of radius $\varepsilon$ around $z$.  By the same argument, after $k\vert X\vert$ steps, the conditional distribution started at $z$ has support on an $L_\infty$ ball (subset of $Z$) of radius $k\varepsilon$.  Because $Z$ is bounded, for sufficiently large $k$, the support of the conditional distribution has measure equal to $\mu(Z)$. 

Suppose Markov Chain 1 only returns to some set $E$ finitely many times in expectation.  Then $E$ is not in the support of the conditional distribution, and must have measure $0$.

\end{proof}

\subsection{Proof of Theorem \ref{thm:symmetric}}

\begin{proof}

Let $\Xi=(X, \phi, \lbrace k_{a,b}\rbrace_{(a,b)\in X})$ be a transaction model on $V$, let $p$ be the transition kernel of Markov Chain 1 using transaction distribution $\Xi$, and let $D$ be some spanning representation.  

It suffices to show that for any measurable set $A$, $\mu(A)=\int_Z p(dz, A)$, where $\mu(\cdot)$ is a uniform measure on $Z$.

A transaction corresponding to a vector $v$ starting from point $z$ lands in $A$ if either $z-v\in A$ -- that is, the transaction succeeds and results in a state in $A$ -- or if $z\in A$ and $z-v\notin Z$ -- that is, the transaction fails and the initial state was already in $A$.  Note also that for any fixed $v$, these events are mutually exclusive.  Hence, the sets $A_{v}^1=\lbrace x : x-v\in A \rbrace \cap Z$ and $A_{v}^2=\lbrace x : x-v\notin Z\rbrace\cap A$ are disjoint.

Note also that $A_{-v}^1=\left\{ x: x+v\in A\right\}\cap Z$.  

Translating by $x\rightarrow x+v$ gives the set $\left\{ y: y\in A, y-v\in Z\right\}$.   Because $\mu$ is uniform on $Z$ and both of these sets are fully contained in $Z$, we get that

\begin{dmath*}
\mu (A_{-v}^1)={\mu(\left\{ y: y\in A, y-v\in Z\right\})}
\end{dmath*}

Note also that the sets $\left\{ y: y\in A, y-v\in Z\right\}$ and $A_v^2$ form a disjoint partition of $A$.  Hence, $\mu(A)=\mu(A_{-v}^1)+\mu(A_v^2)$, and thus 

\begin{dmath}
\label{eqn:mu_iden}
2\mu(A)=\mu(A_{v}^1) + \mu(A_{v}^2) + \mu(A_{-v}^1) + \mu(A_{-v}^2)
\end{dmath}

Expanding terms thus gives

\begin{dmath*}
\int_Z p(dZ, A)={\sum_{(a,b)\in X} \phi(a,b)\int_Z \left( \int_{-\infty}^\infty k_{a,b}(c)\mathds{1}\lbrace z-cD(a,b)\in A\vee (z-cD(a,b)\notin Z \wedge z\in A\rbrace dc\right) dz}
\end{dmath*}

  Changing order of integration, we get 

\begin{dmath*}
\sum_{(a,b)\in X} \phi(a,b) \int_{-\infty}^\infty k_{a,b}(c) \int_Z \mathds{1}{\left\{ (z-cD(a,b)\in A\wedge z\in Z)\vee (z-cD(a,b)\notin Z \wedge z\in A\right\}} dz dc.  
\end{dmath*}



Because $k_{a,b}$ is symmetric, using equation \ref{eqn:mu_iden}, we can rearrange the integral above into 

\begin{equation*}
 \sum_{(a,b)\in X} \phi(a,b) \int_{0}^\infty k_{a,b}(c) (2\mu(A)) dc=\mu(A)
\end{equation*}

 Hence, $\mu$ is stationary in Markov Chain 1.

\end{proof}

\subsection{Proof of Theorem \ref{thm:reliability}}

\begin{proof}

First, note that for a spanning tree $T$, under any spanning representation, the vectors corresponding to $T$ form a basis for $\mathbb{R}^n$.  Moreover, the behavior of a spanning representation on one spanning tree uniquely determines its behavior on all other edges.  Thus, any spanning representation is the result of a linear transformation of another spanning representation.

Let $D_1$ and $D_2$ be two spanning representations, and let $M$ be the (nonsingular) matrix such that for all $e\in E$, $M*D_1(e)=D_2(e)$.  Let $v_1$ be the vector under $D_1$ corresponding to some transaction, and let $v_2=Mv_1$ be vector under $D_2$ for that same transaction.  Let $Z_1$ be the space of represented cycle-equivalent configurations under $D_1$, and let $Z_2=MZ_1$ be the space of represented cycle-equivalent configurations under $D_2$.  Finally, let $\pi_1$ be a density function on $Z_1$ and let $\pi_2$ be the density function on $Z_2$ such that $\pi_2(Mx)=\pi_1(x)$ for $x\in Z_1$.

It suffices to show, therefore, that
\begin{equation*}
 \frac{\int_{Z_1} \pi_1(z)\mathds{1}\left\{ z-v_1\in Z_1\right\} dz}{\int_{Z_1}\pi_1(z)dz}
  =\frac{\int_{Z_2} \pi_2(z)\mathds{1}\left\{ z-v_2\in Z_2\right\} dz}{\int_{Z_2}\pi_2(z)dz}
\end{equation*}

The argument follows from a change of variables.  We note that stretching the representation space should necessarily stretch the density function. Note that
\begin{dmath*}
\int_{Z_2} \pi_2(z)\mathds{1}{\left\{ z-v_2\in Z_2\right\}} dz = \int_{Z_1}\pi_2(Mz)\mathds{1}{\left\{ Mz-Mv_1\in MZ_1\right\}} \vert det(M)\vert dz = \vert det(M)\vert \int_{Z_1} \pi_1(z)\mathds{1}{\left\{ z-v_1\in Z_1\right\}} dz
\end{dmath*} 
and
\begin{dmath*}
\int_{Z_2} \pi_2(z)dz=\int_{Z_1}\pi_2(Mz) \vert det(M)\vert dz=\vert det(M)\vert \int_{Z_1} \pi_1(z)dz
\end{dmath*}

Because $M$ is invertible and $\vert det(M)\vert\neq 0$, the equality holds. 

\end{proof}
\section{Section \ref{sec:monotonicity} Omitted Proofs}
\subsection{Proof of Theorem \ref{thm:zonovolume}}
\begin{proof}

  If $A$ is a set of $n$ vectors in $\mathbb{R}^n$, define $det(A)$ as the determinant of the $n\times n$ matrix where elements of $A$ are interpreted as column vectors.

  By (57) of \cite{shephard1974combinatorial}, the volume of the Minkowski sum of a set of vectors $S=\lbrace x_1,...,x_m\rbrace\subset\mathbb{R}^n$ is equal to $\Sigma_{A\subset S:\vert A\vert=n}\vert det(A)\vert$.  The rest of this proof follows from basic facts about the determinant.

  Let $S$ be the set of vectors corresponding to the directions under a spanning representation $\hat{D}$ of edges in the underlying payment network $G$.  By assumption, $\hat{D}$ maps some spanning tree to the standard basis $B=\lbrace b_1,...,b_n\rbrace$.

  For any linearly independent set $A\subset S$ of size $n$, we would like to show that $\vert det(A)\vert=\vert det(B)\vert=1$.  Note that a set of vectors in $S$ is linearly independent if and only if the set is the set of directions under $\hat{D}$ of a spanning tree in $G$.

  Suppose for induction that the set of vectors $A_i$ is the set of directions of a spanning tree $T_{A_i}$ in $G$, $A_i$ contains $\lbrace b_1,...,b_{i}\rbrace$, $A_i\setminus \lbrace b_1,...,b_{i}\rbrace\subset A$, and $\vert det(A_i)\vert =\vert det(A)\vert $.  Set $A_0=A$.

  If $b_{i+1}\in A_i$, then set $A_{i+1}$ to $A_i$ and the induction hypothesis holds.  Otherwise, viewed as edges in $G$, the set $\lbrace b_{i+1}\rbrace\cup A_i$ contains a single simple cycle.  Say that the vectors corresponding to edges in this cycle are $b_{i+1}, c_1,...,c_k$.  Because $\hat{D}$ is a spanning representation, $b_{i+1}=\Sigma_{j=1}^k\alpha_jc_j$ for $\alpha_j=\pm 1$.  

  Set $A_{i+1}=(A_i\cup\lbrace b_{i+1}\rbrace)\setminus \lbrace c_k\rbrace$.  Then by the linearity of the determinant, we can substitute the column corresponding to $c_k$ with the summation $c_k=1/\alpha_k\Sigma_{j=1}^{k-1}\alpha_jc_j-b_{i+1}$ and expand to get that $det(A_{i})=\Sigma_{j=1}^{k-1}\alpha_j*0 + 1/\alpha_k det(A_{i+1})$.  Hence $\vert det(A_{i+1})\vert=\vert det(A_i)\vert=1$, and the induction hypothesis holds.

  Since $A_n=B$, the above induction shows that $\vert det(A)\vert=\vert det(B)\vert=1$.

  As such, the nonzero elements of the summation $\Sigma_{A\subset S:\vert A\vert=n}\vert det(A)\vert$ correspond exactly to the spanning trees of the underlying graph.  This proves the theorem in the case where every edge has capacity $1$.  The general result follows from the multilinearity of the determinant.  

\end{proof}

\subsection{Proof of Theorem \ref{thm:mkdiffgood}}
\begin{proof}
  This is a general fact about Minkowski summation, but we include here a proof for completeness.

  Let $G$ be the original credit network where $c(x,y)\geq k$, and let $H$ be a credit network where the capacity along $(x,y)$ has been reduced by $k$ (i.e. the capacities of edges in $H$ are $c^\prime$).  Without loss of generality, assume $v$ corresponds to a transaction in the direction along $(x,y)$ that the spanning representation denotes as positive.

  Let $Z$ be the image of the set of configurations of $G$, and let $Z\cap \lbrace {z:z-v\in Z}$ be the image of the set of configurations of $G$ from which transaction $v$ succeeds.  

  Let $Y$ be the image of the set of configurations of $H$ (under the same spanning representation).

  Clearly $Y\subset Z\cap \lbrace {z:z-v\in Z}$ (directly as subsets of $\mathbb{R}^n$).  If $w$ is a configuration of $H$, then $w(x,y)\leq c(x,y)-k\leq c(x,y)$, so the configuration $w^\prime=w$ except $w^\prime(y,x)=w(y,x)+k$ is a configuration of $G$ that maps to the same point in $\mathbb{R}^n$.

  To show that $Z\cap \lbrace {z:z-v\in Z}\subset Y$, it suffices to show that each $z\in Z\cap \lbrace {z:z-v\in Z}$ is the image of some configuration of $G$ where $w(x,y)\leq c(x,y)-k$.  Given such a configuration, we can construct a configuration $w^\prime$ as $w^\prime=w$ except that $w^\prime(y,x)=w(y,x)-k$.  Then $w^\prime$ is a configuration of $H$.

  Let $w$ be some configuration of $G$ from which transaction $v$ is successful.  Then from configuration $w$, $x$ can send at least $k$ units of money to $y$.  If $w(x,y)>c-k$, then there must be some other paths in $w$ that enable $x$ to send $k-(c-w(x,y))$ units of money to $y$.  But then $w$ is cycle-equivalent to a configuration $w^\prime$ constructed from $w$ by sending $k-(c-w(x,y))$ units of money from $x$ to $y$ through paths not using edge $(x,y)$ and then back to $x$ along edge $(y,x)$.

Hence, $Y=Z\cap \lbrace {z:z-v\in Z}$.
\end{proof}

\end{document}